\documentclass[runningheads,11pt]{llncs} 
\usepackage{geometry}
\usepackage{amssymb}
\usepackage{etex} 

\usepackage[cmex10]{amsmath}

\usepackage{ntheorem}
\theoremstyle{break}

\usepackage{graphicx}
\usepackage{verbatim}
\usepackage{tikz}
\usetikzlibrary{arrows,positioning,shadows,matrix,calc}
\usepackage{pgf,pgfarrows,pgfnodes,pgfautomata,pgfheaps,pgfplots}
\usepackage{pst-plot}   

\usepackage[ruled,vlined,linesnumbered]{algorithm2e}

\usepackage{authblk}

\usepackage{wrapfig}

\begin{document}

\mainmatter  

\title{Constructing Privacy Channels from Information Channels}

\author{Genqiang Wu\footnote{This work was partially done when the author was a faculty member of the School of Information Engineering Lanzhou University of Finance and Economics and was a Ph.D candidate of the Institute of Software Chinese Academy of Sciences.}}
\institute{School of Computer Science and Engineering, Chongqing University of Technology 
\\
\email{wahaha2000@icloud.com}
}
\authorrunning{Genqiang Wu}

\maketitle

\begin{abstract}
Data privacy protection studies how to query a dataset while preserving the privacy of individuals whose sensitive information is contained in the dataset. The information privacy model protects the privacy of an individual by using a noisy channel, called privacy channel, to filter out most information of the individual from the query's output. 

This paper studies how to construct privacy channels, which is challenging since  it needs to evaluate the maximal amount of disclosed information  of each individual contained in the query's output, called individual channel capacity. 
Our main contribution is an interesting result which can transform the problem of evaluating a privacy channel's individual channel capacity, which equals the problem of evaluating the capacities of an infinite number of channels, into the problem of evaluating the capacities of a finite number of channels. This result gives us a way to utilize the results in the information theory to construct privacy channels. As some examples, it is used to construct several basic privacy channels, such as the random response privacy channel, the exponential privacy channel and the Gaussian privacy channel, which are respective counterparts of the random response mechanism, the exponential mechanism and the Gaussian mechanism of differential privacy.
\end{abstract}



\section{Introduction} \label{section-introduction}


The field of data privacy protection \cite{DBLP:journals/cacm/Dwork11,DBLP:journals/csur/FungWCY10,DBLP:series/ads/2008-34} studies how to query a dataset while preserving the privacy of individuals whose sensitive information is contained in the dataset. 
Since data privacy is a fundamental problem in today's information age, this field becomes more and more important.  
The differential privacy model \cite{DBLP:conf/tcc/DworkMNS06,DBLP:conf/icalp/Dwork06} is currently the most important and popular data privacy protection model and has obtained great success both in the computer science communities \cite{DBLP:journals/fttcs/DworkR14,DBLP:books/sp/17/Vadhan17} and in the industry companies \cite{DBLP:conf/ccs/ErlingssonPK14,DBLP:journals/corr/abs-1709-02753}.  
We reason that its success is strongly related to the feature that its definition of privacy is irrelevant to the semantic meaning of data, which significantly separates it from many other privacy models, such as the $k$-anonymity \cite{DBLP:journals/ijufks/Sweene02}, the $\ell$-diversity \cite{DBLP:conf/icde/MachanavajjhalaGKV06} or the $t$-closeness \cite{DBLP:conf/icde/LiLV07} model etc. Note that this feature is very similar to the one of Shannon's information theory which strips  away the ``semantic aspects'' from the definition of information \cite{Shannon1948,sep-information-semantic}. The other important feature of the differential privacy model is the composition privacy property, which makes it possible to modularly construct complex differentially private algorithms from simple ones.

In spite of the big success of the differential privacy model,  after a decade or so of study, researchers, however, found that this model is vulnerable to those attacks where adversaries have knowledge of the dependence among different records in the queried dataset \cite{DBLP:conf/sigmod/KiferM11}.  Furthermore, it seems that this model is not so flexible to provide satisfactory data utility \cite{DBLP:conf/crypto/GehrkeHLP12}, especially for the unstructured data \cite{DBLP:conf/tcc/KasiviswanathanNRS13,DBLP:conf/sigmod/ChenZ13}.
Therefore, many new privacy models are introduced to alleviate the above two defects of the differential privacy model, such as the crowd-blending privacy model \cite{DBLP:conf/crypto/GehrkeHLP12}, the Pufferfish model \cite{DBLP:journals/tods/KiferM14}, the coupled-worlds privacy model \cite{DBLP:conf/focs/BassilyGKS13}, the zero-knowledge privacy model \cite{DBLP:conf/tcc/GehrkeLP11}, the inferential privacy model \cite{DBLP:conf/innovations/GhoshK16} and the information privacy model \cite{wu-he-xia2018}, etc. 
These models  in general are either more  secure or more flexible in utility than the differential privacy model. Especially, they all inherit the first feature of differential privacy mentioned in the preceding paragraph. 
However, it seems that the second important feature of differential privacy, i.e. the composition privacy property, is very hard to be inherited by these privacy models.

We are more interested in the information privacy model, a variant of Shannon's information-theoretic model of the encryption systems \cite{6769090}, since it has many interesting properties as the differential privacy model, such as the post-processing invariance property \cite{wu-he-xia2018}, the group privacy property and the composition privacy property \cite{DBLP:journals/corr/abs-1907-09311}. 
Informally speaking, the information privacy model protects privacy by using a noisy channel, called privacy channel or privacy mechanism, to filter out most information of each individual from the query's output. 
This model also assumes that the entropy of each adversary's probability distribution to the queried dataset is not so small.
This assumption is used to improve the data utility of query's output, which is reasonable when the queried dataset is big enough. 

In this paper we consider how to construct privacy channels. Especially, we want to construct some basic ones, such as the Laplace privacy channel, the Gaussian privacy channel and the exponential privacy channel which should be the respective counterparts of the Laplace mechanism, the Gaussian mechanism and the exponential mechanism of the differential privacy model. The great attraction of these basic privacy channels is that by combining them with the composition and post-processing invariance properties, we can \emph{modularly} construct other complex privacy channels as in the differential privacy model.  
Although there is the great attraction, constructing privacy channels are challenging works
 since it is related to evaluate the maximal amount of information of each individual, called individual channel capacity of the privacy channel, which can be obtained by an adversary from the query's output. 
 We stress that the individual channel capacity in general is hard to be evaluated since it needs to evaluate the capacities of \emph{infinitely many} channels, which seems to be intractable.    
 
\textbf{Our Contribution:}  The main contribution of this paper is an interesting result which can transform the problem of evaluating a privacy channel's individual channel capacity, which is related to the problem of evaluating the capacities of an \emph{infinite number} of channels, into the problem of evaluating the capacities of a \emph{finite number} of channels. This result is very useful to utilize the results in the information theory to construct privacy channels. As some examples, we construct several useful privacy channels by using it, such as the random response privacy channel, the exponential privacy channel and the Gaussian privacy channel, which as mentioned above are respective counterparts of the random response mechanism, the exponential mechanism and the Gaussian mechanism of the differential privacy model \cite{DBLP:journals/fttcs/DworkR14}.

\textbf{Outline:} The following part of this paper is organized as follows. Section \ref{section:prelim} presents the information privacy model and shows that the differential privacy model is a special case of the former. Section \ref{sec-individual-channel-capacity} introduces a result which can transform the problem of evaluating the individual channel capacity to the problem of evaluating the capacities of some channels in the information theory. In Section \ref{sec:some-channels} we show how to use the result in Section \ref{sec-individual-channel-capacity} to construct privacy channels, such as the random response privacy channel and the exponential privacy channel. In Section  \ref{section-continuous-case} we extend the result in Section \ref{sec-individual-channel-capacity} to the continuous case and then construct the Gaussian privacy channel. Section \ref{sec:illustration-dp} give an interesting illustration of the differential privacy mechanism using the probability transition matrix of privacy channel. Section \ref{section:conclusion} concludes this paper.

\section{Preliminaries} \label{section:prelim}

In this section, we restate the information privacy model introduced in \cite{wu-he-xia2018,DBLP:journals/corr/abs-1907-09311}. As mentioned in Section \ref{section-introduction}, this model is a variant of Shannon's information theoretic model of the encryption systems \cite{6769090}. 
As a comparative study, we also show that the differential privacy model is equivalent to a special kind of the information privacy model, which implies that the two models are compatible with each other. 
Most notations of this paper follow the book \cite{DBLP:books/daglib/0016881}.

\subsection{The Information Privacy Model}

In the information privacy model there are $n\ge 1$ individuals and a dataset has $n$ records, each of which is the private data of  one individual. 
Let the random variables $X_1, \ldots, X_n$ denote an adversary's probabilities/uncertainties to the $n$ records in the queried dataset. Let $\mathcal X_i$ denote the record universe of $X_i$.
A dataset $x:=(x_1,\ldots,x_n)$ is a sequence of $n$ records, where each $x_i$ is an assignment of $X_i$ and so $x_i\in \mathcal X_i$.
Let $\mathcal X = \prod_{i\in [n]} \mathcal X_i$ where $[n]=\{1, \ldots, n\}$. Let $\mathbb P$ and $\mathbb P_i$ denote the universe of probability distributions over $\mathcal X$ and $\mathcal X_i$, respectively. Let $p(x)$ and $p(x_i)$ denote the probability distribution of $X$ and $X_i$, respectively. 
Let $\Delta$ be a subset of $\mathbb P$. 
Furthermore, we sometimes abuse a capital letter, such as $X$, to either denote a random variable or denote the probability distribution which the random variable follows.  
Then if the probability distribution of the random variable $X$ is in $\Delta$, we say that $X$ is in $\Delta$, denoted as $X\in\Delta$. 
Moreover, if $X=(X_{I_1},X_{I_2})$ follows the joint probability distribution $p(x_{I_1})p(x_{I_2}|x_{I_1})$ and $X \in \Delta$, we denote $p(x_{I_1}) \in \Delta$ and $p(x_{I_2}|x_{I_1})\in \Delta$.

\begin{definition}[The Knowledge of an Adversary]
Let the random vector $X:=(X_1, \ldots, X_n)$ denote the uncertainties/probabilities of an adversary to the queried dataset. Then $X$ or its probability distribution is called the knowledge of the adversary (to the dataset).
\end{definition}


In order to achieve more data utility, the information privacy model restricts adversaries' knowledges. 
Note that, by letting all adversaries' knowledges be derived from a subset $\Delta$ of $\mathbb P$, we achieve a restriction to adversaries' knowledges.
In this paper we assume that the entropy of each adversary's knowledge is not so small, which is formalized as the following assumption.

\begin{assumption} \label{assumption-1}
Let $b$ be a positive constant. Then, for any one adversary's knowledge $X$, there must be $X\in \mathbb P_b$, where
\begin{align}
\mathbb P_b = \{X: H(X)\ge b\}
\end{align}
with $H(X)$ being the entropy of $X$. 
\end{assumption}

For a query function $f$ over $\mathcal X$, let $\mathcal Y= \{f(x): x\in \mathcal X\}$ be its range. Except in Section \ref{section-continuous-case}, this paper assumes that the sets $\mathcal X, \mathcal Y$ both are finite. We now define the privacy channel/mechanism.
\begin{definition}[Privacy Channel/Mechanism]
To the function $f:\mathcal X \rightarrow \mathcal Y$,
we define a \emph{privacy channel/mechanism}, denoted by $(\mathcal X, p(y|x), \mathcal Y)$ or $p(y|x)$ for simplicity (when there is no ambiguity), to be the probability transition matrix $p(y|x)$ that expresses the probability of observing the output symbol $y\in \mathcal Y$ given that we query the dataset $x\in \mathcal X$.  
\end{definition}

To the above privacy channel $p(y|x)$, let $Y$ be its output random variable. We now define the individual channel capacity which is used to model the largest amount of information of an individual that an adversary can obtain from the output $Y$.  

\begin{definition}[Individual Channel Capacity]
To the function $f:\mathcal X \rightarrow \mathcal Y$ and its one privacy channel $p(y|x) $, we define the individual channel capacity of $p(y|x) $ with respect to $\Delta \subseteq \mathbb P$ as
\begin{align}
C_1=\max_{i\in [n],X\in \Delta} I(X_i;Y),
\end{align}
where $X$ equals $(X_1,\ldots,X_n)$ and $I(X_i;Y)$ is the mutual information between $X_i$ and $Y$.
\end{definition}

The information privacy model protects privacy by restricting the individual channel capacity. 

\begin{definition}[Information Privacy] \label{definition-ip}
To the function $f:\mathcal X \rightarrow \mathcal Y$, we say that its one privacy channel $p(y|x) $ satisfies $\epsilon$-information privacy with respect to $\Delta$ if 
\begin{align}
C_1 \le \epsilon,
\end{align}
where $C_1$ is the individual channel capacity of $p(y|x)$ with respect to $\Delta$.
\end{definition}

If we set $\Delta=\mathbb P_b$ and assume that Assumption \ref{assumption-1} is true, then letting the channel $p(y|x)$ satisfy $\epsilon$-information privacy with respect to $\Delta$ will ensure that 
any adversary can only obtain at most $\epsilon$ bits information of each individual from the output of the channel.

We now define the (utility-privacy) balance function of a privacy channel. This function is used to express the ability of the privacy channel to preserve the privacy when adversaries' knowledges are reduced to $\mathbb P_b$ from $\mathbb P$ in order to improve the data utility. 

\begin{definition}[Balance Function]  \label{definition-2}              
To the privacy channel $p(y|x)$, let $C_1, C_1^{b}$ be its individual channel capacities with respect to $\mathbb P, \mathbb P_b$, respectively. To each fixed $b$, assume there exists a nonnegative constant $\delta$ such that
\begin{align}
C_1 =C_1^b +\delta.
\end{align}
Then we say that the privacy channel is $(b,\delta)$-(utility-privacy) balanced, and that the function 
\begin{align}
\delta=\delta(b), \hspace{1cm} b\in [0,\log |\mathcal X|]
\end{align}
is the (utility-privacy) balance function of the privacy channel.
\end{definition}

The balance function is an increasing function since, when $b\le b'$, there is $\mathbb P_b \supseteq \mathbb P_{b'}$, which implies $C_1^b \ge C_1^{b'}$ and then $\delta(b) \le \delta(b')$. 

\begin{lemma}[Monotonicity of Balance Function] \label{lemma-4}
Let $\delta=\delta(b)$ be the balance function of the privacy channel $p(y|x)$. Then this function is non-decreasing and therefore there is
\begin{align}
0=\delta(0)\le \delta(b)\le \delta(\log|\mathcal X|) < \min \left\{ b,\max_{i\in [n]}\log|\mathcal X_i| \right\}.
\end{align}
\end{lemma}

Clearly, the more larger $b$ is means that the more weaker adversaries are, which then implies the more larger data utility and the more larger $\delta$. Therefore, the parameter $b$ can be considered as an indicator of the data utility but the parameter $\delta$ can be considered as an indicator of the amount of the private information lost when the data utility has $b$ amount of increment.

\begin{lemma}[Theorem 1 in \cite{DBLP:journals/corr/abs-1907-09311}] \label{lemma-5}
Let $\delta=\delta(b)$ be the balance function of the privacy channel $p(y|x)$.
Then the privacy channel $p(y|x)$ satisfies $\epsilon$-information privacy with respect to $\mathbb P_b$ if and only if it satisfies $(\epsilon+\delta)$-information privacy with respect to $\mathbb P$.
\end{lemma}

Lemma \ref{lemma-5} shows that, to construct an $\epsilon$-information privacy channel with respect to $\mathbb P_b$, we only need to construct an $(\epsilon+\delta)$-information privacy channel with respect to $\mathbb P$, where $\delta=\delta(b)$ is the balance function of the channel. 
However, we find that it is very hard to evaluate the balance functions of privacy channels, which is left as an open problem.
Of course, as explained in \cite{DBLP:journals/corr/abs-1907-09311}, even though we don't know the balance functions, this doesn't mean that we can't use them. Specifically, we can first set the needed value of $\delta$ (but not $b$), then there must \emph{exist} a largest $b=b'$ such that $\delta(b') \le \delta$ by the monotonicity of balance functions (even though we can't find it). 
Then we can freely use Lemma \ref{lemma-5} to construct $\epsilon$-information privacy channels with respect to $\mathbb P_{b'}$ from the $(\epsilon+\delta)$-information privacy channels with respect to $\mathbb P$. The only drawback of this way is that we don't know the value of $b'$ and then can't estimate the extent of reasonability of Assumption \ref{assumption-1} for $b=b'$.

\subsection{The Differential Privacy Model} \label{subsec-dp-model}

We stress that the differential privacy model is compatible with the information privacy model. We will define a special case of the information privacy models and then show that it is equivalent to the differential privacy model. 
We first define the max-mutual information, which is the maximum of $\log\frac{p(y|x)}{p(y)}$. Note that the mutual information $I(X;Y)$ is the expected value of $\log\frac{p(y|x)}{p(y)}$.  

\begin{definition}[Max-Mutual Information \cite{DBLP:conf/nips/DworkFHPRR15,DBLP:conf/focs/RogersRST16}] \label{definition-max-MI}
The max-mutual information of the random variables $X, Y$ is defined as
\begin{align}
I_{\infty}(X;Y)= \max_{x\in \mathcal X , y\in \mathcal Y}\log\frac{p(y|x)}{p(y)}.
\end{align}
\end{definition}

Similarly, we can define the max-individual channel capacity.

\begin{definition}[Max-Individual Channel Capacity]
To the function $f:\mathcal X \rightarrow \mathcal Y$ and its one privacy channel $p(y|x) $, we define the max-individual channel capacity of $p(y|x) $ with respect to $\Delta \subseteq \mathbb P$ as
\begin{align}
C_1^{\infty}=\max_{i\in [n],X\in \Delta}  I_{\infty}(X_i;Y),
\end{align}
where $X=(X_1,\ldots,X_n)$.
\end{definition}

We now present the max-information privacy which restricts the max-individual channel capacity, and in which the adversaries' knowledges are restricted to the independent knowledges to each records in the queried dataset. We will also give the definition of differential privacy and then show that they are equivalent to each other. 

\begin{definition}[Max-Information Privacy]
A privacy channel $p(y|x)$ satisfies $\epsilon$-max-information privacy with respect to $\mathbb P_{ind}$ if it satisfies 
\begin{align}
C_1^{\infty} \le \epsilon,
\end{align}
where $C_1^{\infty}$ is the max-invididual channel capacity of the privacy channel $p(y|x)$ with respect to $\mathbb P_{ind}$, where 
\begin{align}
\mathbb P_{ind} = \left\{X\in \mathbb P: p(x) = \prod_{i=1}^np(x_i), \forall x=(x_1, \ldots, x_n)\in \mathcal X\right\},
\end{align}
with $X \sim p(x)$ and $X_i \sim p(x_i)$. Note that $\mathbb P_{ind}$ is the universe of $X$ such that $X_1, \ldots, X_n$ are independent to each other. 
\end{definition}


\begin{definition}[Differential Privacy \cite{DBLP:conf/tcc/DworkMNS06,DBLP:journals/fttcs/DworkR14}] \label{definition-1}
A privacy channel $p(y|x)$ satisfies $\epsilon$-differential privacy if, for any $y\in \mathcal Y$ and any $x,x'\in \mathcal X$ with $\|x-x'\|_1\le 1$, there is
\begin{align}
\log \frac{p(y|x)}{ p(y|x')} \le \epsilon.
\end{align}
\end{definition}

The following lemma shows that the above two apparently different definitions, in fact, are equal. 

\begin{proposition} \label{proposition-1}
The privacy channel $p(y|x)$ satisfies $\epsilon$-max-information privacy with respect to $\mathbb P_{ind}$ if and only if it satisfies $\epsilon$-differential privacy. 
\end{proposition}

\begin{proof} 
This is a direct corollary of Theorem 1 in \cite{wu-he-xia2018}. 
\qed
\end{proof}

Proposition \ref{proposition-1} has two implications. First, the differential privacy model is founded on the independence assumption, i.e. founded on the assumption that $X \in \mathbb P_{ind}$ for all possible $X$, which is too strong to be satisfied than Assumption \ref{assumption-1}. Therefore, the information privacy model with respect to $\mathbb P_b$ is more reasonable than the differential privacy model. Second, the information privacy model is compatible with the differential privacy model. This gives us large facility to study the information privacy model by comparing it with the differential privacy model.

\section{From Individual Channel Capacity to Channel Capacity} \label{sec-individual-channel-capacity}

The preceding section shows that the information privacy model is compatible with the differential privacy model. Then a natural question is: \emph{Are the Laplace channel, the Gaussian channel and the exponential channel also applicable in the information privacy model? } Before answering this question, we first should be able to evaluate the individual channel capacities of the corresponding privacy channels, which are new problems both in the data privacy protection and in the information theory.  
In this section we will develop approaches to solve these problems. 


\subsection{A Simple Example} \label{subsec-simple-example}

It is very complicated to evaluate the individual channel capacity of a privacy channel. Before given the general results, a simple example is suitable to explain our ideas. 

\begin{example} \label{example-1}
Set $\mathcal X_1=\{0,1,2\}$, $\mathcal X_2=\{0,1\}$ and $\mathcal X=\mathcal X_1 \times \mathcal X_2$. The function $f:\mathcal X \rightarrow \mathcal Y$ is defined as
\begin{align}
f(x_1,x_2) =  
	\begin{cases} 
		1 &  x_1= x_2 \\
		0 & \mbox{otherwise}
	\end{cases}
\end{align}
for $(x_1,x_2) \in \mathcal X$. Then $\mathcal Y=\{0,1\}$. The probability transition matrix of the privacy channel $p(y|x)$ is shown in Fig. \ref{figure-1},
\begin{figure}
\centering
\begin{tikzpicture}[every node/.style={anchor=north east,fill=white,minimum width=1.4cm,minimum height=7mm}]
\matrix (mA) [draw,matrix of math nodes]
{
p(y_2|x_{11},x_{22}) & (y_2|p(x_{12},x_{22})  & p(y_2|x_{13},x_{22})  \\
p(y_2|x_{11},x_{21}) & (y_2|p(x_{12},x_{21})  & p(y_2|x_{13},x_{21})  \\
};

\matrix (mB) [draw,matrix of math nodes] at ($(mA.south west)+(3.5,-0.5)$)
{
p(y_1|x_{11},x_{22}) & (y_1|p(x_{12},x_{22})  & p(y_1|x_{13},x_{22})  \\
p(y_1|x_{11},x_{21}) & (y_1|p(x_{12},x_{21})  & p(y_1|x_{13},x_{21})  \\
};

\draw[-](mA.north east)--(mB.north east);
\draw[-](mA.north west)--(mB.north west);
\draw[-](mA.south east)--(mB.south east);
\draw[->]($(mB.south west)-(0.3,0.3)$)--($(mB.south east)-(0.3,0.3)$)node [near end, below right] {$x_1$};
\draw[->]($(mB.south west)-(0.3,0.3)$)--($(mB.north west)-(0.3,0.3)$)node [near end, below left] {$x_2$};
\draw[->]($(mB.north west)-(0.3,-0.3)$)--($(mA.north west)-(0.3,-0.3)$)node [near end, above left] {$y$};
\end{tikzpicture}
\caption{The probability transition matrix of the privacy channel $p(y|x)$ in Example \ref{example-1}, which is a three-dimensional matrix} \label{figure-1}
\end{figure}
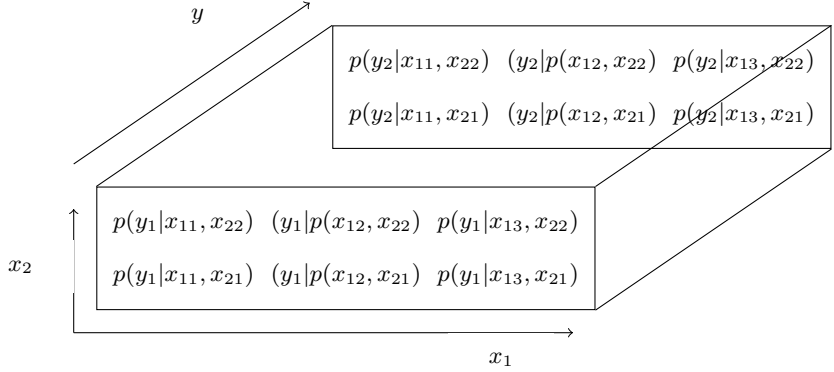
where $(x_{11},x_{12},x_{13}) = (0,1,2)$, $(x_{21},x_{22}) = (0,1)$ and $(y_1,y_2) =(0,1)$.  
\end{example} 

The individual channel capacity of the privacy channel $p(y|x)$ can be simplified as
\begin{align}
C_1=\max_{i\in [n],X\in \mathbb P} I(X_i;Y) =\max_{i\in [n]} \max_{p(y|x_i) \in \mathbb Q_i} \max_{p(x_i) \in \mathbb P_i}  I(X_i;Y),
\end{align}
where $\mathbb P_i$ denotes the universe of the probability distributions of $X_i$, and $\mathbb Q_i$ denotes the universe of the probability transition matrices $p(y|x_i)$ such that $X\in \mathbb P$. Then the main problem becomes to evaluate what is the set $\mathbb Q_i$. 
Recall that $p(y|x_i) =\sum_{x_{(i)}} p(y|x_i,x_{(i)})p(x_{(i)}|x_i)$ for $(y,x_i) \in \mathcal Y \times \mathcal X_i$, where $(i) = 2$ if $i=1$ and $(i)=1$ if $i=2$. 
We first consider the case of $i=1$. There is 
\begin{align} \label{equation-16}
\begin{pmatrix}
p(y_1|x_1) \\
p(y_2|x_1)
\end{pmatrix}
= 
\begin{pmatrix}
p(y_1|x_1,x_{21}) & p(y_1|x_1, x_{22}) \\
p(y_2|x_1,x_{21}) & p(y_2|x_1, x_{22}) 
\end{pmatrix}
\begin{pmatrix}
p(x_{21}|x_1) \\
p(x_{22}|x_1)
\end{pmatrix}
\end{align}
for $x_1 \in \mathcal X_1=\{x_{11},x_{12},x_{13}\}$.  Note that every column vector in the above equation is a probability distribution. Let $A \in \mathbb Q_1$ and assume
\begin{align} \label{equation-17}
A=
\begin{pmatrix}
p(y_1|x_{11}) & p(y_1|x_{12})  & p(y_1|x_{13})\\
p(y_2|x_{11}) & p(y_2|x_{12})  & p(y_2|x_{13})
\end{pmatrix}
\end{align}
Then the $j$th column vector of $A$ is a convex combination of the column vectors of the first matrix in the right side of the equation (\ref{equation-16}) when setting $x_1=x_{1j}$, $j=\{1,2,3\}$. By combining the above result with the fact that the first matrix in the right side of the equation (\ref{equation-16}),  when setting $x_1=x_{1j}$, is the $j$th plane, i.e. a two-dimensional matrix, in Fig. \ref{figure-1} that is parallel to the $x_2$-$y$ plane,  we conclude that \emph{the $j$th column vector of the matrix $A \in \mathbb Q_1$ is a convex combination of the column vectors of the $j$th plane paralleled to the $x_2$-$y$ plane in Fig. \ref{figure-1}}. 
Furthermore, since $X\in \mathbb P$, then the probability distribution $(p(x_{21}|x_1), p(x_{22}|x_1) )$ over $\mathcal X_2$ can be any possible binary distribution, which implies that the coefficients of the convex combination can be any binary convex coefficients. Therefore, \emph{the set $\mathbb Q_1$ contains and only contains those matrices $A$, whose $j$th column vector is a convex combination of the column vectors of the $j$th plane paralleled to the $x_2$-$y$ plane.} Specifically, we have
\begin{align}
\mathbb Q_1 = \{ ( \mathcal X_1, p(y|x_1), \mathcal Y):  p(x_{2}|x_1) \in \mathbb P\},
\end{align}
where $ p(x_{2}|x_1) \in \mathbb P$ means there exists $X \in \mathbb P$ such that $X \sim p(x_1) p(x_{2}|x_1)$. 
Similarly,  \emph{the set $\mathbb Q_2$ contains and only contains those matrices $A$, whose $j$th column vector is a convex combination of the column vectors of the $j$th plane paralleled to the $x_1$-$y$ plane.} That is, 
\begin{align}
\mathbb Q_2 = \{ ( \mathcal X_2, p(y|x_2), \mathcal Y):  p(x_{1}|x_2) \in \mathbb P\}, 
\end{align}
where $ p(x_{1}|x_2) \in \mathbb P$ means there exists $X \in \mathbb P$ such that $X \sim p(x_2) p(x_{1}|x_2)$. 

The above discussions tell us that each set $\mathbb Q_i$ is infinite and it is, therefore, not a reasonable way to evaluate $C_1$ by evaluating the channel capacity of each channel in the sets $\mathbb Q_i, i\in [n]$. Now we construct a new set $\mathbb S_i$, which is finite and is a small subset of $\mathbb Q_i$,  and show that it has the same maximal channel capacity as $\mathbb Q_i$. Let the subset $\mathbb S_1$ ($\mathbb S_2$) of $\mathbb Q_1$ ($\mathbb Q_2$) contains and only contains those matrices whose $j$th column vector is one column vector of the $j$th plane paralleled to the $x_2$-$y$ ($x_1$-$y$) plane.
Formally, set
\begin{align}
\mathbb S_1
=& \left\{( \mathcal X_1, p(y|x_1), \mathcal Y): p(x_{2}|x_1) \in \mathbb P_{2}^{dg} \right\} \\
=& \left\{( \mathcal X_1, p(y|x_1), \mathcal Y): \mbox{for each } x_1 \in \mathcal X_1, \exists x_{2}\in \mathcal X_2 \mbox{ such that } \forall y, p(y|x_1)=p(y|x_1,x_{2}) \right\}
\end{align}
and
\begin{align}
\mathbb S_2
=& \left\{( \mathcal X_2, p(y|x_{2}), \mathcal Y): p(x_{1}|x_2) \in \mathbb P_{1}^{dg} \right\}\\
=& \left\{( \mathcal X_2, p(y|x_2), \mathcal Y): \mbox{for each } x_2\in \mathcal X_2, \exists x_{1}\in \mathcal X_1 \mbox{ such that } \forall y, p(y|x_2)=p(y|x_1,x_{2}) \right\},
\end{align}
where $\mathbb P_{2}^{dg}$ is the universe of the probability transition matrix $p(x_{2}|x_1)$ whose each column vector is one degenerate probability distribution over $\mathcal X_2$, and where $\mathbb P_{1}^{dg}$ is the universe of the probability transition matrix $p(x_{1}|x_2)$ whose each column vector is one degenerate probability distribution over $\mathcal X_1$.
Then the $j$th column vector of a matrix in $ \mathbb Q_i$ is a convex combination of those $j$th column vectors of the matrices in $\mathbb S_i$. 
Therefore, we have 
\begin{align} \label{equation-28}
C_1=\max_{i\in [n],X\in \mathbb P} I(X_i;Y) =\max_{i\in [n]} \max_{p(y|x_i) \in \mathbb Q_i} \max_{p(x_i) \in \mathbb P_i}  I(X_i;Y)= \max_{i\in [n]} \max_{p(y|x_i) \in \mathbb S_i} \max_{p(x_i) \in \mathbb P_i}  I(X_i;Y)
\end{align}
by the convexity of the mutual information in Lemma \ref{lemma-6}.

\begin{lemma}[Theorem 2.7.4 in \cite{DBLP:books/daglib/0016881}] \label{lemma-6}
Let $(X,Y) \sim p(x,y) = p(x)p(y|x)$. Then the mutual information $I(X;Y)$ is a convex function of $p(y|x)$ for fixed $p(x)$. 
\end{lemma}

\subsection{The General Results} \label{subsec-general-result}

Now it's time to give the generalized version of the results in Section \ref{subsec-simple-example}. Set $\mathcal X_{(i)} = \prod_{j \ne i} \mathcal X_j$, $X_{(i)} = (X_1, \ldots, X_{i-1}, X_{i+1}, \ldots, X_n)$, and  $x_{(i)} = (x_1, \ldots, x_{i-1}, x_{i+1}, \ldots, x_n)$.
Let $p(y|x)$ be one privacy channel of $f:\mathcal X \rightarrow \mathcal Y$. Then the probability transition matrix of $p(y|x)$ is an $(n+1)$-dimensional matrix, one of whose three-dimensional special case is shown in Fig. \ref{figure-1}. 
By the equation $p(y|x_i) = \sum_{x_{(i)}} p(y|x_1,x_{(i)})p(x_{(i)}|x_i)$, to each $x_i\in \mathcal X_i = \{x_{i1}, \ldots, x_{i|\mathcal X_i|}\}$, we have 
\begin{align} \label{equation-8}
\begin{pmatrix}
p(y_1|x_i)\\
p(y_2|x_i) \\
\vdots \\
p(y_m|x_i)
\end{pmatrix} 
=
\begin{pmatrix}
p(y_1|x_{(i)1},x_i)           & p(y_1|x_{(i)2},x_i) & \cdots    & p(y_1|x_{(i)|\mathcal X_{(i)}|},x_i) \\
p(y_2|x_{(i)1},x_i)           & p(y_2|x_{(i)2},x_i) & \cdots    & p(y_2|x_{(i)|\mathcal X_{(i)}|},x_i) \\
\vdots                               & \vdots                      &\ddots     &         \vdots                                           \\
p(y_m|x_{(i)1},x_i)           & p(y_m|x_{(i)2},x_i) & \cdots    & p(y_m|x_{(i)|\mathcal X_{(i)}|},x_i) 
\end{pmatrix}
\begin{pmatrix}
p(x_{(i)1}|x_i) \\
p(x_{(i)2}|x_i)  \\
\vdots \\
p(x_{(i)|\mathcal X_{(i)}|}|x_i) 
\end{pmatrix}
\end{align}
which is a generalization of the equation (\ref{equation-16}), 
where $\mathcal X_{(i)} = \{ x_{(i)1},x_{(i)2}, \ldots, x_{(i)|\mathcal X_{(i)}|} \}$. 
Leting $A \in \mathbb Q_i$, then $A$ is of the following form
\begin{align} \label{equation-9}
A=
\begin{pmatrix}
p(y_1|x_{i1}) & p(y_1|x_{i2})  & \cdots & p(y_1|x_{i|\mathcal X_i|}) \\
p(y_2|x_{i1}) & p(y_2|x_{i2})  & \cdots & p(y_2|x_{i|\mathcal X_i|}) \\
\vdots                               & \vdots                      &\ddots     &         \vdots   \\
p(y_m|x_{i1}) & p(y_m|x_{i2})  & \cdots & p(y_m|x_{i|\mathcal X_i|}) 
\end{pmatrix}
\end{align}
which is a generalization of (\ref{equation-17}).  So, the $j$th column vector of $A $ is a convex combination of the column vectors of the $j$th hyperplane paralleled to the $x_1$-$\ldots$-$x_{i-1}$-$x_{i+1}$-$\ldots$-$x_{n}$-$y$ hyperplane.\footnote{The $x_1$-$\ldots$-$x_{i-1}$-$x_{i+1}$-$\ldots$-$x_{n}$-$y$ hyperplane is the hyperplane extended by the $x_1$, \ldots, $x_{i-1}$, $x_{i+1}$, \ldots, $x_{n}$ and $y$  coordinate axises. This hyperplane is the generalization of the $x_2$-$y$ plane in Fig. \ref{figure-1}.} Then \emph{the set $\mathbb Q_i$ contains and only contains those matrices $A$, whose $j$th column vector is a convex combination of the column vectors of the 
$j$th hyperplane paralleled to the $x_1$-$\ldots$-$x_{i-1}$-$x_{i+1}$-$\ldots$-$x_{n}$-$y$ hyperplane.} More formally, we have
\begin{align} \label{equation-24}
\mathbb Q_i = \{ ( \mathcal X_i, p(y|x_i), \mathcal Y):  p(x_{(i)}|x_i) \in \mathbb P\},
\end{align}
where $ p(x_{(i)}|x_i) \in \mathbb P$ means there exists $X \in \mathbb P$ such that $X \sim p(x_i) p(x_{(i)}|x_i)$. 
Similarly, we define that \emph{the set $\mathbb S_i$ contains and only contains those matrix $A$, whose $j$th column vector is one column vector of the 
$j$th hyperplane paralleled to the $x_1$-$\ldots$-$x_{i-1}$-$x_{i+1}$-$\ldots$-$x_{n}$-$y$ hyperplane.} Formally, we define
\begin{align} 
\mathbb S_i 
\label{equation-25}
=& \left\{( \mathcal X_i, p(y|x_i), \mathcal Y): p(x_{(i)}|x_i) \in \mathbb P_{(i)}^{dg} \right\}\\
=& \left\{( \mathcal X_i, p(y|x_i), \mathcal Y): \mbox{for each } x_i, \exists x_{(i)} \mbox{ such that } \forall y, p(y|x_i)=p(y|x_i,x_{(i)}) \right\},
\end{align}
where $\mathbb P_{(i)}^{dg}$ is the universe of the probability transition matrix $p(x_{(i)}|x_i)$ whose each column vector is one degenerate probability distribution over $\mathcal X_{(i)}$.
Combining Lemma \ref{lemma-6} with the above results, we then have the following theorem.

\begin{theorem} \label{theorem-1}
Let $\mathbb Q_i$ and $\mathbb S_i$ be as shown in (\ref{equation-24}) and (\ref{equation-25}), respectively. Then $\mathbb S_i \subseteq \mathbb Q_i$ and 
each channel in $\mathbb Q_i$ is one convex combination of the channels in $\mathbb S_i$. Therefore, there is
\begin{align}
C_1=\max_{i\in [n],X\in \mathbb P} I(X_i;Y) =\max_{i\in [n]} \max_{p(y|x_i) \in \mathbb Q_i} \max_{p(x_i) \in \mathbb P_i}  I(X_i;Y)= \max_{i\in [n]} \max_{p(y|x_i) \in \mathbb S_i} \max_{p(x_i) \in \mathbb P_i}  I(X_i;Y),
\end{align}
where $C_1$ is the individual channel capacity of the privacy channel $p(y|x)$ with respect to $\mathbb P$.
\end{theorem}

Theorem \ref{theorem-1} gives us a way to utilize the results in the information theory to construct privacy channels. We will show its powerful function in the next section.

\section{Some Basic Privacy Channels} \label{sec:some-channels}

In this section we show how to construct privacy channels by using Theorem \ref{theorem-1}. The key point is to transplant typical  channels in the information theory into the information privacy model to construct privacy channels. Some privacy channels are of special interest as discussed in the beginning of Section \ref{sec-individual-channel-capacity}, such as the random response channel, the Gaussian channel and the  (discrete) exponential channel which are respective counterparts of the random response mechanism, the Gaussian mechanism and the exponential mechanism of the differential privacy model. 

\subsection{Binary Symmetric/Randomized Response Privacy Channel} \label{subsec-binary-channel}

The binary symmetric privacy channel is a generalization of either the binary symmetric channel in \cite[\S  8.1.4]{DBLP:books/daglib/0016881} or the randomized response mechanism in the differential privacy model \cite{DBLP:journals/fttcs/DworkR14}. Let the query function $f$ be defined as $f: \mathcal X \rightarrow \mathcal Y=\{0,1\}$, such as the one in Example \ref{example-1}. Then the \emph{binary symmetric privacy channel} or the \emph{randomized response privacy channel} of $f$ is defined as
\begin{align} \label{equation-18}
p(y|x) = (1-p)^{1-|y-f(x)|}p^{|y-f(x)|}, \hspace{1cm} x\in \mathcal X,  y\in \mathcal Y,
\end{align}
where the output $f(x)$ is complemented with probability $p \in (0,1)$; that is, an error occurs with probability $p$ and in this condition, a $0$ of the value of $f(x)$ is received as a $1$ and vice versa. This privacy channel is the simplest privacy channel; yet it captures most of the complexity of the general problem.

We now evaluate the individual channel capacity of the binary symmetric privacy channel. By Theorem \ref{theorem-1}, we have 
\begin{align}
C_1 =& \max_{i\in [n]} \max_{p(y|x_i) \in \mathbb S_i} \max_{p(x_i) \in \mathbb P_i}  I(X_i;Y) \\
=& \max_{i\in [n]} \max_{p(y|x_i) \in \mathbb S_i} \max_{p(x_i) \in \mathbb P_i} \left( H(Y) - \sum_{x_i} p(x_i)H(Y|X=x_i) \right) \\
=&\max_{i\in [n]} \max_{p(y|x_i) \in \mathbb S_i} \max_{p(x_i) \in \mathbb P_i} \left( H(Y) - \sum_{x_i} p(x_i)H(p)\right)  \label{equation-10}\\
=&\max_{i\in [n]} \max_{p(y|x_i) \in \mathbb S_i} \max_{p(x_i) \in \mathbb P_i} \left( H(Y) - H(p)\right) \\
\le & \log 2- H(p)
\end{align}
where (\ref{equation-10}) follows from that the conditional probability distribution of $Y|X=x_i$ is one line parallel to the $y$-coordinate axis in Fig. \ref{figure-1} which is just the one in (\ref{equation-18}), where the last inequality follows since $Y$ is a binary random variable, and where 
\begin{align}
H(p) = -p\log p -(1-p)\log (1-p).
\end{align} 
\emph{The key point of the above evaluating is that each channel in the set $\mathbb S_i$ is a binary symmetric channel in the information theory.} 

Setting $\log 2-H(p) \le \epsilon$ with $\epsilon \in (0,1)$, then $H(p) \ge \log 2-\epsilon$. Therefore, when setting
\begin{align}
p\in (p^*,1-p^*),
\end{align} 
the binary symmetric privacy channel satisfies $\epsilon$-information privacy with respect to $\mathbb P$, where $p^*$ is the solution of the equation
\begin{align}
H(p) = \log 2-\epsilon.
\end{align}

\subsection{Data-Independent Privacy Channels}

We now consider the data-independent privacy channels \cite{DBLP:conf/sigmod/HayMMCZ16,DBLP:journals/isci/Soria-ComasD13}, of which the binary symmetric privacy channel is one special case.  The data-independent privacy channels are very popular and important in the differential privacy model. For example, the Laplace mechanism and the Gaussian mechanism both are (continuous-case) data-independent privacy channels. Informally speaking, for a data-independent channel, each probability distribution is constant across $\mathcal Y$. The following gives the definition. 

\begin{definition}[Data-Independent Privacy Channel]
Let $p(y|x)$ be one privacy channel of the function $f:\mathcal X \rightarrow \mathcal Y$. For any two datasets $x,x'\in \mathcal X$, if the two probability distributions $p(y|x), p(y|x')$ over $\mathcal Y$ are a permutation of each other, then the privacy channel is said to be data-independent.
\end{definition}

Now we estimate the individual channel capacities of the data-independent privacy channels. Let the privacy channel $p(y|x)$ be data-independent. Then there is
\begin{align}
C_1=& \max_{i\in [n]} \max_{p(y|x_i) \in \mathbb S_i} \max_{p(x_i) \in \mathbb P_i}  I(X_i;Y) \\
=& \max_{i\in [n]} \max_{p(y|x_i) \in \mathbb S_i} \max_{p(x_i) \in \mathbb P_i}\left( H(Y) - \sum_{x_i} p(x_i)H(Y|X=x_i) \right) \\
=&\max_{i\in [n]} \max_{p(y|x_i) \in \mathbb S_i} \max_{p(x_i) \in \mathbb P_i}\left( H(Y) - \sum_{x_i} p(x_i)H(Z)\right)  \\
=&\max_{i\in [n]} \max_{p(y|x_i) \in \mathbb S_i} \max_{p(x_i) \in \mathbb P_i} \left( H(Y) - H(Z)\right) \\
=&\max_{i\in [n]} \max_{p(y|x_i) \in \mathbb S_i} \max_{p(x_i) \in \mathbb P_i} H(Y) - H(Z) \\
\le & \log |\mathcal Y|- H(Z),  \label{equation-14}
\end{align}
where the random variable $Z \sim p(y|x)$ for a fixed $x\in \mathcal X$. The same as in Section \ref{subsec-binary-channel}, the key point of the above proof is that each channel in $\mathbb S_i$ is a data-independent channel in the information theory. 
Furthermore, in general, the upper bound in (\ref{equation-14}) is hard to be reached. The following defined privacy channel, however, is an exception.

\begin{definition}[Weakly Symmetric(-Reachable) Privacy Channel] \label{definition-4}
Let $p(y|x)$ be a data-independent privacy channel of $f:\mathcal X\rightarrow \mathcal Y$ and let $\mathbb S_i$ be defined as in Section \ref{subsec-general-result}. Then the privacy channel is said to be weakly symmetric(-reachable) if there exist one $i'\in [n]$ and one channel $p^*(y|x_{i'})\in \mathbb S_{i'}$ such that $\sum_{x_{i'}\in \mathcal X_{i'}} p^*(y|x_{i'})$ is constant for all $y\in \mathcal Y$. 
\end{definition}

The weakly symmetric privacy channel can reach the upper bound in (\ref{equation-14}), whose proof is the same as \cite[Theorem 8.2.1]{DBLP:books/daglib/0016881}. We then have the following result.

\begin{corollary}
Let $C_1$ be the individual channel capacity of the data-independent privacy channel $p(y|x)$ of $f:\mathcal X\rightarrow \mathcal Y$ with respect to $\mathbb P$. Then there is
\begin{align}
C_1 \le \log |\mathcal Y| - H(Z),
\end{align}
where $Z \sim p(y|x)$ for a fixed $x\in \mathcal X$. Furthermore, if the privacy channel $p(y|x)$ is weakly symmetric, then there is 
\begin{align}
C_1 = \log |\mathcal Y| - H(Z),
\end{align}
which is achieved by a uniform distribution on the set $\mathcal X_{i'}$, where ${i'}$ is the one in Definition \ref{definition-4}.
\end{corollary}

\subsection{(Discrete) Exponential Privacy Channel}

The exponential mechanism \cite{DBLP:conf/focs/McSherryT07,DBLP:journals/fttcs/DworkR14} is fundamental in the differential privacy model. 
Now we construct its counterpart in the information privacy model. 
Before presenting the channel, we have to do some preparation. 

\begin{definition}[Distortion Function \cite{DBLP:books/daglib/0016881,wu-he-xia2017-b}]
A distortion function is a mapping
\begin{align}
d:\mathcal Y \times \mathcal Y \rightarrow \mathbb R^+,
\end{align}
where $\mathbb R^+$ is the set of nonnegative real numbers.
The distortion $d(y,y')$ is a measure of the cost of representing the symbol $y$ by the symbol $y'$.
\end{definition}

For the function $f:\mathcal X \rightarrow \mathcal Y$, set $\mathcal Y = \{y_1, y_2, \ldots, y_k\}$. Then $|\mathcal Y|=k$. To each $x\in \mathcal X$, assume
\begin{align}
d(y_{i_1},f(x)) \le d(y_{i_2},f(x)) \le \cdots \le d(y_{i_k},f(x)),
\end{align}
where $(i_1, i_2,\ldots, i_k)$ is one permutation of $[k]$. Then define a function $\phi_x$ on $\mathcal Y$ as 
\begin{align} \label{equation-12}
\phi_x(y_{i_j}) = j-1, \hspace{1cm} j\in [k].
\end{align}

\begin{definition}[(Discrete) Exponential Privacy Channel]
To each $x\in \mathcal X$, we define 
\begin{align} \label{equation-13}
p(y|x) = \frac{1}{\alpha} \exp\left(\frac{-\phi_x(y)}{N}\right), \hspace{1cm}y\in \mathcal Y,
\end{align}
where $\phi_x$ is defined as in (\ref{equation-12}), $N$ is a positive parameter, and
\begin{align}
\alpha= \sum_y \exp\left(\frac{-\phi_x(y)}{N}\right)=\sum_{i=0}^{k-1}\exp\left(\frac{-i}{N}\right)
\end{align}
is the normalizer. Then the above defined privacy channel $p(y|x)$ is said to be the (discrete) exponential privacy channel with the parameter $N$.
\end{definition}

Note that the exponential privacy channel defined above is a little different from the exponential mechanism in \cite{DBLP:conf/focs/McSherryT07,DBLP:journals/fttcs/DworkR14}. This is intentioned since we want to let the privacy channel be data-independent, which is easy to be verified.  Let $Z \sim p(y|x)$ where the probability distribution $p(y|x)$ is defined in (\ref{equation-13}). We now evaluate the entropy $H(Z)$ of $Z$. Recall that $k=|\mathcal Y|$. 
We have
\begin{align}
H(Z) =-\sum_y p(y|x)\log p(y|x) =& \frac{1}{\alpha}\sum_{y} \exp\left(\frac{-\phi_x(y)}{N}\right) \left (\frac{\phi_x(y)}{N}+\log \alpha \right)\\
=&\frac{1}{\alpha}\sum_{i=0}^{k-1} \exp\left(\frac{-i}{N}\right) \left(\frac{i}{N}+\log \alpha\right) \\
=& \log \alpha + \frac{1}{\alpha}\sum_{i=0}^{k-1} \exp\left(\frac{-i}{N}\right) \frac{i}{N} 
\end{align}
Note that the last expression is related to the two series $\alpha, \sum_{i=0}^{k-1} \exp\left(\frac{-i}{N}\right) \frac{i}{N} $.  
The series $\alpha$ is geometric and then is evaluated as
\begin{align}
\alpha= \frac{1-\exp(-k/N)}{1-\exp(-1/N)}. 
\end{align}
The second series $\sum_{i=0}^{k-1} \exp\left(\frac{-i}{N}\right) \frac{i}{N} $ is a bit complicated. Set $\lambda=1/N$. There is
\begin{align}
  \sum_{i=0}^{k-1}  \exp\left(\frac{-i}{N}\right) \frac{i}{N} 
=& \lambda \left( \sum_{i=0}^{k-1}  \exp\left(-i\lambda\right) (i+1) - \sum_{i=0}^{k-1}  \exp\left(-i\lambda\right)  \right) \\
=& \lambda \left( \exp(\lambda)\left(\sum_{i=0}^{k-1}  \exp\left(-i\lambda\right) i  + k\exp(-k\lambda)\right) -\sum_{i=0}^{k-1}  \exp\left(-i\lambda\right) \right) \\
=& \lambda \left( \exp(\lambda)\sum_{i=0}^{k-1}  \exp\left(-i\lambda\right) i  + k\exp((1-k)\lambda) - \frac{1-\exp(-k\lambda)}{1-\exp(-\lambda)}\right)
\end{align}
Then the second series have the value
\begin{align}
 \sum_{i=0}^{k-1}  \exp\left(-i\lambda\right) i =& \frac{k\exp((1-k)\lambda) - \frac{1-\exp(-k\lambda)}{1-\exp(-\lambda)}}{1-\exp(\lambda)} \\
 =& k\frac{\exp((1-k)\lambda) }{1-\exp(\lambda)} -\frac{ 1-\exp(-k\lambda)}{(1-\exp(-\lambda))(1-\exp(\lambda))} 
\end{align}
Therefore,
\begin{align}
H(Z) &= \log \alpha + \frac{1}{\alpha}\sum_{i=0}^{k-1} \exp\left(\frac{-i}{N}\right) \frac{i}{N} \\
=& \log \frac{1-\exp(-k\lambda)}{1-\exp(-\lambda)}+ \frac{1-\exp(-\lambda)}{1-\exp(-k\lambda)} \lambda\left(  k\frac{\exp((1-k)\lambda) }{1-\exp(\lambda)} -\frac{ 1-\exp(-k\lambda)}{(1-\exp(-\lambda))(1-\exp(\lambda))} \right) \\
=&\log \frac{1-\exp(-k\lambda)}{1-\exp(-\lambda)}+ \frac{\lambda}{\exp(\lambda)-1}  -  \frac{k \lambda}{\exp(k\lambda)-1}
\end{align}
where $k \ge 1$ and $\lambda=\frac{1}{N} \le 1$.

Let $\log k - H(Z) \le \epsilon$, which ensures $C_1 \le \epsilon$ by (\ref{equation-14}). Then
\begin{align}
H(Z) =\log \frac{1-\exp(-k\lambda)}{1-\exp(-\lambda)}+ \frac{\lambda}{\exp(\lambda)-1}  -  \frac{k \lambda}{\exp(k\lambda)-1}\ge \log k -\epsilon.
\end{align}
Using numerical method we can find suitable $\lambda=1/N$ or $N$ to satisfy the above inequality, which then ensures $\epsilon$-information privacy with respect to $\mathbb P$ of the exponential privacy channel.


\section{Continuous Case} \label{section-continuous-case}

In this section we consider the case where $\mathcal X_i, i\in [n]$  are continuous sets. In this setting, $X_i$ and $X$ are continuous variables. We also let $Y$ and $\mathcal Y$ be continuous. Then all the related probability distributions of the above random variables become the corresponding density functions. The entropy and the mutual information become integrals as defined in \cite[\S 9]{DBLP:books/daglib/0016881}; that is,
\begin{align}
H(X) = - \int_x p(x)\log p(x)dx
\end{align}
and 
\begin{align} \label{equation-19}
I(X_i;Y) = \int_{x_i,y} p(x_i,y) \log\frac{p(x_i,y)}{p(x_i)p(y)}dx_idy= \int_{x_i,y} p(x_i)p(y|x_i) \log\frac{p(y|x_i)}{\int_{x_i'}p(y|x_i')p(x_i')dx_i'}dx_idy,
\end{align}  
where $p(y|x_i) = \int_{x_{(i)}} p(y|x_i,x_{(i)}) p(x_{(i)}|x_i) dx_{(i)}$. 

The main task of this section is to generalize Theorem \ref{theorem-1} to the continuous case. We now define the sets $\mathbb Q_i$ and $\mathbb S_i$. Recall that the channel $(\mathcal X_i, p(y|x_i), \mathcal Y)$ has of the formula
\begin{align} 
p(y|x_i) = \int_{x_{(i)}} p(y|x_i,x_{(i)}) p(x_{(i)}|x_i) dx_{(i)}, \hspace{1cm} (x_i, y) \in (\mathcal X_i, \mathcal Y). 
\end{align}
Let $\mathbb Q_i$ be the set of all possible above channels; that is, let
\begin{align} \label{equation-30}
\mathbb Q_i = \left\{ ( \mathcal X_i, p(y|x_i), \mathcal Y):  p(x_{(i)}|x_i) \in \mathbb P \right\},
\end{align}
where $ p(x_{(i)}|x_i) \in \mathbb P$ means there exists $X \in \mathbb P$ such that $X \sim p(x_i) p(x_{(i)}|x_i)$. Let $\mathbb S_i$ be the set
\begin{align} \label{equation-29}
\mathbb S_i 
=& \left\{( \mathcal X_i, p(y|x_i), \mathcal Y): \mbox{for each } x_i, \exists x_{(i)} \mbox{ such that } \forall y, p(y|x_i)=p(y|x_i,x_{(i)}) \right\}.
\end{align}
Then the continuous case generalization of Theorem \ref{theorem-1} is shown as follow.  

\begin{theorem} \label{corollary-1}
Assume the individual channel capacity $C_1$ of the privacy channel $(\mathcal X, p(y|x), \mathcal Y)$ with respect to $\mathbb P$, defined as
\begin{align}
C_1=\max_{i\in [n],X\in \mathbb P} I(X_i;Y)=\max_{i\in [n]} \max_{p(y|x_i) \in \mathbb Q_i} \max_{p(x_i) \in \mathbb P_i}  I(X_i;Y), 
\end{align}
exists, and let $\mathbb Q_i$ and $\mathbb S_i$ be defined in (\ref{equation-30}) and (\ref{equation-29}), respectively. Assume the function $p(y|x)$ is continuous on $\mathcal Y \times \mathcal X$. 
Then there are $\mathbb S_i \subseteq \mathbb Q_i$ and  
\begin{align} \label{equation-40}
C_1=\max_{i\in [n]} \max_{p(y|x_i) \in \mathbb S_i} \max_{p(x_i) \in \mathbb P_i}  I(X_i;Y).
\end{align}
\end{theorem}

The proof of the above theorem needs the continuous case generalization of Lemma \ref{lemma-6}.

\begin{lemma}\label{lemma-7}
Let $X,Y$ both be continuous random variables. 
Let $(X,Y) \sim p(x,y) = p(x)p(y|x)$. Then the mutual information $I(X;Y)$ is a convex function of $p(y|x)$ for fixed $p(x)$. 
\end{lemma}

\begin{proof}
The proof is omitted since it is the same as the one of  \cite[Theorem 2.7.4]{DBLP:books/daglib/0016881}. 
\qed
\end{proof}

Now we prove Theorem \ref{corollary-1}.

\begin{proof}[of Theorem \ref{corollary-1}]
We first do a simplification. We have
\begin{align}  \label{equation-35}
p(y|x_i) = \int_{x_{(i)}} p(y|x_i,x_{(i)}) p(x_{(i)}|x_i) dx_{(i)}=\int_{x_{(i)}} p(y|x_i,x_{(i)})  \mu^{x_i}(dx_{(i)}),
\end{align}
where $\mu^{x_i}$ is the probability measure on $\mathcal X_{(i)}$ such that
\begin{align}
\mu^{x_i}(A) = \int_{x_{(i)}\in A} p(x_{(i)}|x_i) dx_{(i)}
\end{align}
for all measurable set $A\subseteq \mathcal X_{(i)}$. Clearly, there is $\mu^{x_i}(\mathcal X_{(i)})=1$. 
Note that the probability measures $\mu^{x_i}$ in (\ref{equation-35}) are different for different $x_i\in \mathcal X_i$. 

The first claim $\mathbb S_i \subseteq \mathbb Q_i$ is proved as follows. 
Let $p'(y|x_i) \in \mathbb S_i$. Then, to each $x_i$, there exists $x_{(i)}^{x_i}\in \mathcal X_{(i)}$ such that 
\begin{align}
p'(y|x_i) = p(y|x_i, x_{(i)}^{x_i})
\end{align}  
for all $y\in \mathcal Y$ by (\ref{equation-29}). In order to prove the first claim, we need to prove that there exists a probability transition matrix $\mu=\{\mu^{x_i}: x_i\in \mathcal X_i\}\in \mathbb P$ such that 
\begin{align}
( \mathcal X_i, p(y|x_i), \mathcal Y)_{\mu} = ( \mathcal X_i, p'(y|x_i), \mathcal Y),
\end{align}
where the left side of the equation is the channel when setting the probability transition matrix $p(x_{(i)}|x_i)$ to be $\mu$.
This probability transition matrix $\mu$ is constructed as the limit of a sequence of probability transition matrix $\{\mu_k\}_{k\ge 1}$ as follows.
Fix $x_i$ and $k$. By (\ref{equation-35}), set $p(x_{(i)}|x_i)$ to be the simple function \cite[Definition 2.3.1]{athreya2006measure}
\begin{align}
p(x_{(i)}|x_i) = p^{x_i}_k(x_{(i)}) =
\begin{cases}
c_k^{x_i} & x_{(i)} \in A_k^{x_i} \\
0    & x_{(i)} \in A_1^{x_i}=\mathcal X_{(i)}-A_k^{x_i}
\end{cases}
\end{align}
with $\int_{x_{(i)}\in \mathcal X_{(i)}}p(x_{(i)}|x_i)dx_{(i)} = c_k^{x_i} \mu_k^{x_i}(A_k^{x_i}) =1$, where 
\begin{align}
A_k^{x_i}= B_{1/k}(x_{(i)}^{x_i})= \left\{x_{(i)}\in \mathcal X_{(i)}: \|x_{(i)}-x_{(i)}^{x_i}\|_2\le 1/k\right\}
\end{align} 
is the ball with the center $x_{(i)}^{x_i}$ and the radius $1/k$. 
Inputting the above probability transition matrix $\mu_k=\{\mu_k^{x_i}: x_i\in \mathcal X_i\}$ into (\ref{equation-35}) and letting $k \rightarrow \infty$, we have 
\begin{align} 
p(y|x_i) 
=& \lim_{k \rightarrow \infty} \int_{x_{(i)}\in A_k^{x_i}} p(y|x_i,x_{(i)}) c_k^{x_i}\mu_k^{x_i}(dx_{(i)}) \\ \label{equation-36}
=& \lim_{k \rightarrow \infty} \frac{\int_{x_{(i)}\in A_k^{x_i}} p(y|x_i,x_{(i)})  \mu_k^{x_i}(dx_{(i)})}{\mu_k^{x_i}(A_k^{x_i})} \\ \label{equation-37}
=&  \lim_{k \rightarrow \infty}  p(y|x_i, \bar x_{(i)}^{x_i})   \hspace{1cm} \exists \bar x_{(i)}^{x_i} \in A_k^{x_i}\\ \label{equation-38}
=& p(y|x_i,  x_{(i)}^{x_i})\\
=&p'(y|x_i), 
\end{align}
where (\ref{equation-36}) is due to $\mu_k^{x_i}(A_k^{x_i})c_k^{x_i}=1$, (\ref{equation-37}) is due to the mean value theorem, and (\ref{equation-38}) is due to the continuity of the function $p(y|x)$ on $\mathcal Y \times \mathcal X$. Since $\mu=\lim_{k\rightarrow \infty}\mu_k \in \mathbb P$, we then have $p'(y|x) \in \mathbb Q_i$, which implies $\mathbb S_i \subseteq \mathbb Q_i$. The first claim is proved. 

The second claim, i.e. the equation (\ref{equation-40}), is proved as follows. Since the first claim holds, then we only need to prove
\begin{align} \label{equation-39}
C_1\le\max_{i\in [n]} \max_{p(y|x_i) \in \mathbb S_i} \max_{p(x_i) \in \mathbb P_i}  I(X_i;Y).
\end{align}
The crux of the following  proofs is the observation that an integral of a function is defined as a limitation of the integrals of a sequence of simple functions \cite[\S 2]{athreya2006measure}. Specifically, the integral in (\ref{equation-35}) is a limitation of the integrals of a sequence of simple functions, which will be defined as follows.
Then the inequality (\ref{equation-39}) should be a corollary of the convexity of mutual information in Lemma \ref{lemma-7}.
To each $x_i$, let 
\begin{align}
p_{k_{x_i}} ^{x_i}(x_{(i)})= \sum_{j=1}^{k_{x_i}} c_j^{x_i} I_{A_j^{x_i}} (x_{(i)}), \hspace{.5cm}x_{(i)}\in \mathcal X_{(i)}
\end{align}
be a simple function on $\mathcal X_{(i)}$ as defined in \cite[Definition 2.3.1]{athreya2006measure}, and assume $p_{k_{x_i}} ^{x_i}(x_{(i)}) \uparrow_{k_{x_i}} p(y|x_i,x_{(i)}) $ for all $x_{(i)}\in \mathcal X_{(i)}$. Then, to each $x_i$, there is 
\begin{align}  \label{equation-31}
p(y|x_i) =  \lim_{k_{x_i} \rightarrow \infty} \int_{x_{(i)}} p_{k_{x_i}} ^{x_i}(x_{(i)})\mu(dx_{(i)}) = \lim_{k_{x_i} \rightarrow \infty}  \sum_{j=1}^{k_{x_i}} c_j^{x_i} \mu(A_j^{x_i}) 
\end{align}
by \cite[Definition 2.3.2, Definition 2.3.3]{athreya2006measure}. Inputting (\ref{equation-31}) into (\ref{equation-19}), we have
\begin{align}
I(X_i;Y) 
=& \int_{x_i,y} p(x_i)\lim_{k_{x_i} \rightarrow \infty}  \sum_{j=1}^{k_{x_i}} c_j^{x_i} \mu(A_j^{x_i})  \log\frac{\lim_{k_{x_i} \rightarrow \infty}  \sum_{j=1}^{k_{x_i}} c_j^{x_i} \mu(A_j^{x_i}) }{\int_{x_i'}\lim_{k_{x_i'} \rightarrow \infty}  \sum_{j=1}^{k_{x_i'}} c_j^{x_i'} \mu(A_j^{x_i'}) p(x_i')dx_i'}dx_idy \\
=& \lim_{k_{x_i},k_{x_i'} \rightarrow \infty:\forall x_i,x_i'}
\int_{x_i,y} p(x_i) \sum_{j=1}^{k_{x_i}} c_j^{x_i} \mu(A_j^{x_i})  \log\frac{ \sum_{j=1}^{k_{x_i}} c_j^{x_i} \mu(A_j^{x_i}) }{ \int_{x_i'} \sum_{j=1}^{k_{x_i'}} c_j^{x_i'} \mu(A_j^{x_i'}) p(x_i')dx_i'}dx_idy \\ \label{equation-32}
\le& \lim_{k_{x_i},k_{x_i'} \rightarrow \infty:\forall x_i,x_i'}\sum_{j=1}^{k_{x_i}}  \mu(A_j^{x_i})
\int_{x_i,y} p(x_i) c_j^{x_i}  \log\frac{  c_j^{x_i}  }{ \int_{x_i'}  c_j^{x_i'}  p(x_i')dx_i'}dx_idy \\ \label{equation-33}
\le & \lim_{k_{x_i}\rightarrow \infty:\forall x_i} \max_{j\in [k_{x_i}]} \int_{x_i,y} p(x_i) c_j^{x_i} \log\frac{c_j^{x_i} }{\int_{x_i'}  p(x_i') c_j^{x_i'} dx_i'}dx_idy  \\ \label{equation-34}
\le& \int_{x_i,y} p(x_i) p(y|x_i,  x_{(i)}^{x_i}) \log\frac{p(y|x_i,  x_{(i)}^{x_i}) }{\int_{x_i'} p(y|x_i,  x_{(i)}^{x_i'}) p(x_i')dx_i'}dx_idy 
\end{align} 
where (\ref{equation-32}) is due to the convexity of mutual information in Lemma \ref{lemma-7}, (\ref{equation-33}) is due to 
$\sum_{j=1}^{k_{x_i}}\mu(A_j^{x_i})=1$,
and (\ref{equation-34}) is due to the continuity of the function $p(y|x)$ on $\mathcal Y\times \mathcal X$ which ensures
\begin{align}
\lim_{k_{x_i}\rightarrow \infty:\forall x_i} c_j^{x_i} = p(y|x_i,  x_{(i)}^{x_i}) 
\end{align} 
for one $x_{(i)}^{x_i} \in \mathcal X_{(i)}$. Therefore, the inequality (\ref{equation-39}) holds 
since the channel $p(y|x_i)=p(y|x_i,  x_{(i)}^{x_i})$ in the right side of (\ref{equation-34}) is in $\mathbb S_i$. 

The claims are proved. 
\qed
\end{proof}

\subsection{Gaussian Privacy Channel}

Now we give the Gaussian privacy channel, which is the counterpart of the Gaussian mechanism of differential privacy \cite[\S A]{DBLP:journals/fttcs/DworkR14}. We will not discuss the Laplace privacy channel, the counterpart of the Laplace mechanism of differential privacy. The reason why we choose the Gaussian channel instead of the Laplace channel is that the study of the former is more mature in the information theory. 

\begin{definition}[Gaussian Privacy Channel] \label{definition-3}
The privacy channel $p(y|x)$ of the continuous function $f:\mathcal X \rightarrow \mathcal Y$ is said to be Gaussian with the parameter $N$ if, to each $x\in \mathcal X$, there is
\begin{align}
p(y|x) \sim \mathcal N(f(x), N),
\end{align}
where $\mathcal Y\subseteq \mathbb R$.
\end{definition}

Now we evaluate the individual channel capacity of the Gaussian privacy channel.
The key point of this evaluating is that each channel $p(y|x_i)$ in $\mathbb S_i$ is a Gaussian channel in \cite[\S 10.1]{DBLP:books/daglib/0016881}. Note that the continuity of the function $f(x)$ on $\mathcal X$ ensures the continuity of the function $p(y|x)$ on $\mathcal Y \times \mathcal X$.
Let $\mathcal Y \subseteq [-T,T] $ with $T >0$. Then $Ef(X)^2 \le T^2$ for all $X\in \mathbb P$ and 
\begin{align}
C_1
=&  \max_{i\in [n]} \max_{p(y|x_i) \in \mathbb Q_i} \max_{p(x_i) \in \mathbb P_i}  I(X_i;Y) \\
=&  \max_{i\in [n]} \max_{p(y|x_i) \in \mathbb S_i} \max_{p(x_i) \in \mathbb P_i}  \left( H(Y) - \sum_{x_i}p(x_i)H(Y|x_i) \right) \\
=& \max_{i\in [n]} \max_{p(y|x_i) \in \mathbb S_i} \max_{p(x_i) \in \mathbb P_i}   H(Y) - H(Z)  \\
\le &  \frac{1}{2} \log \left(1+\frac{T^2}{N}\right) \label{equation-26}
\end{align}
where $Z \sim \mathcal N(0, N)$, and where (\ref{equation-26}) is due to \cite[\S 10.1]{DBLP:books/daglib/0016881}. 
Let 
\begin{align}
\frac{1}{2}\log(1+\frac{T^2}{N})\le \epsilon.
\end{align}
Then 
\begin{align}
N\ge \frac{T^2}{\exp(2\epsilon)-1}.
\end{align}
Therefore, if $N\ge \frac{T^2}{\exp(2\epsilon)-1}$, then the Gaussian privacy channel in Definition \ref{definition-3} satisfies $\epsilon$-information privacy with respect to $\mathbb P$.

We stress that the above result is only suitable to the case where $f(x) \in \mathbb R$. To the multi-dimensional query function $f: \mathcal X \rightarrow \mathbb R^d$, we can either use the Gaussian privacy channel in each dimension of $f(x)$ or construct multi-dimensional Gaussian privacy channel by using the results of \cite[Theorem 9.4.1]{DBLP:books/daglib/0016881} and \cite[Theorem 9.6.5]{DBLP:books/daglib/0016881}. The details are omitted in this paper.

\subsection{Basic Utility Analysis}

Now we give an intuitive comparison of the data utilities between the Gaussian privacy channel in Definition \ref{definition-3} and the Gaussian mechanism, the Laplace mechanism in the differential privacy model.
Noting that the variances of the probability distributions in these channels/mechanisms are crucial to indicate the utility, therefore, we only compare their variances.   

The variances of the Laplace mechanism and the Gaussian mechanism are respective $\frac{\Delta f}{\epsilon}$ and 
\begin{align} \label{equation-27}
\sqrt{2\ln (1.25/\delta')} \frac{\Delta f}{\epsilon},
\end{align}
in which (\ref{equation-27}) ensures $(\epsilon,\delta')$-differential privacy and $\Delta f$ is the global sensitivity. The variance of the Gaussian privacy channel in Definition \ref{definition-3} is 
\begin{align}
\sqrt{N} \ge  \frac{T}{\sqrt{\exp(2\epsilon)-1}}  \approx_{\epsilon \rightarrow 0} \frac{T}{\sqrt{2\epsilon}}.
\end{align}
Noticing that $\Delta \le T $ for the same $\epsilon$ and for the moderate $\delta'$, then the variances of the Laplace mechanism and the Gaussian mechanism are more smaller than the one of the Gaussian privacy channel. This implies that the first two mechanisms, in general, have more better utilities than the Gaussian privacy channel when $b=0$.  However, if the queried dataset is big enough, i.e. if $n$ is large enough, which ensures that Assumption \ref{assumption-1} is true for a moderately large $b$ and then for a moderately large $\delta$ (in Definition \ref{definition-2}), then the variance of the Gaussian privacy channel becomes
\begin{align}
\sqrt{N} \ge  \frac{T}{\sqrt{\exp(2\epsilon+2\delta)-1}}  \approx_{\epsilon + \delta\rightarrow 0} \frac{T}{\sqrt{2\epsilon+2\delta}}.
\end{align}
by Lemma \ref{lemma-5}. Then in this setting the Gaussian privacy channel may have more smaller variance than the Laplace mechanism and the Gaussian mechanism, especially when $n$ is large enough. This implies that the first channel roughly have more better utility than the last two mechanisms when $n$ is large enough.

The above discussions indicate that the information privacy model is more flexible than the differential privacy model and therefore may obtain more better utility than the later when $n$ is large enough.

\section{An Illustration of Differential Privacy} \label{sec:illustration-dp}

Recall that the differential privacy model can be considered as a special case of the information privacy model as discussed in Section \ref{subsec-dp-model}.
As a by-product, we now use the probability transition matrix of a privacy channel to illustrate the differentially private mechanisms. This
illustration can be used to explain the approaches introduced in \cite{wu-he-xia2017-b} to construct differentially private mechanisms. 


\subsection{Local Model ($n=1$)}

The local model \cite{DBLP:conf/ccs/ErlingssonPK14,DBLP:journals/fttcs/DworkR14} obtains more and more attentions. In this model, each individual first perturbs its own data such that the results satisfy $\epsilon$-differential privacy and then submits the results to the aggregator. This is equivalent to the case that there is only one record in the queried dataset and then $n=1$. 
Then the inequalities in Definition \ref{definition-1} can be given a graphical explanation from the privacy channel's probability transition matrix
\begin{align} \label{equation-22}
\begin{pmatrix}
p(y_1|x_1) & p(y_1|x_2) & \cdots & p(y_1|x_{\ell}) \\
p(y_2|x_1) & p(y_2|x_2) & \cdots & p(y_2|x_{\ell}) \\
\vdots & \vdots & \ddots & \vdots \\
p(y_k|x_1) & p(y_k|x_2) & \cdots & p(y_k|x_{\ell}) 
\end{pmatrix}
\end{align}
where $\mathcal X=\{x_1,\ldots, x_{\ell}\}, \mathcal Y=\{y_1,\ldots,y_k\}$. Specifically, the inequalities in Definition \ref{definition-1} is equivalent to the constraint that \emph{any two elements in each row of the matrix (\ref{equation-22})  have the ratio $\le e^{\epsilon}$.} 

\subsection{The Case of $n=2$}

In this case, we only give an example to explain the results. The general results are shown in Section \ref{subsec-dp-general-result}. 
We reuse the function $f$ in Example \ref{example-1}, which is just the case of $n=2$. As shown in Fig. \ref{figure-1}, its privacy channel is a three-dimensional matrix. Then the inequalities of differential privacy in Definition \ref{definition-1} is equivalent to the constraint that \emph{any two elements in each line paralleled either to the  $x_1$-coordinate axis or to the $x_2$-coordinate axis in Fig. \ref{figure-1} have the ratio $\le e^{\epsilon}$}.  

\subsection{The General Result} \label{subsec-dp-general-result}

Let $A$ be the probability transition matrix of the differentially private channel $(\mathcal X, p(y|x), \mathcal Y)$. Then it is an $(n+1)$-dimensional matrix and can be shown in an (n+1)-dimensional space as in Fig. \ref{figure-1}. Then the inequalities in Definition \ref{definition-1} is equivalent to the constraint that \emph{any two elements in each line of the matrix $A$, paralleled to one of the $x_i$-coordinate axis for $i\in [n]$,  have the ratio $\le e^{\epsilon}$.}

\section{Conclusion}\label{section:conclusion}

This paper shows that Theorem \ref{theorem-1} and Theorem \ref{corollary-1} are powerful approaches to transplant the results of  channels in the information theory into the information privacy model to construct privacy channels. For example, the binary symmetric privacy channel and the Gaussian privacy channel are constructed by using the channels in the information theory. This implies a promising direction to construct privacy channels using the channels in the information theory.

The Gaussian privacy channel and the exponential privacy channel can take the role of the Laplace mechanism and the exponential mechanism of differential privacy. Specifically, by combining the above privacy channels with the post-processing invariance property, i.e. the data-processing inequality, the group privacy property and the composition privacy property \cite{DBLP:journals/corr/abs-1907-09311}, we are ready to modularly construct complicated privacy channels for complicated data processing problems in the information privacy model. 

The utilities comparison of the Gaussian privacy channel of this paper and the Gaussian mechanism, the Laplace mechanism of the differential privacy model gives us a glimpse of the information privacy model: The information privacy model roughly is more flexible to obtain more better utility than the differential privacy model without loss of privacy when the queried dataset is large enough.  

How to evaluate the balance functions of  privacy channels is a more complicated problem, which is left as an open problem.

\bibliographystyle{unsrt}

\begin{thebibliography}{10}

\bibitem{DBLP:journals/cacm/Dwork11}
Cynthia Dwork.
\newblock A firm foundation for private data analysis.
\newblock {\em Commun. ACM}, 54(1):86--95, 2011.

\bibitem{DBLP:journals/csur/FungWCY10}
Benjamin C.~M. Fung, Ke~Wang, Rui Chen, and Philip~S. Yu.
\newblock Privacy-preserving data publishing: {A} survey of recent
  developments.
\newblock {\em {ACM} Comput. Surv.}, 42(4), 2010.

\bibitem{DBLP:series/ads/2008-34}
Charu~C. Aggarwal and Philip~S. Yu, editors.
\newblock {\em Privacy-Preserving Data Mining - Models and Algorithms},
  volume~34 of {\em Advances in Database Systems}.
\newblock Springer, 2008.

\bibitem{DBLP:conf/tcc/DworkMNS06}
Cynthia Dwork, Frank McSherry, Kobbi Nissim, and Adam~D. Smith.
\newblock Calibrating noise to sensitivity in private data analysis.
\newblock In {\em Theory of Cryptography, Third Theory of Cryptography
  Conference, {TCC} 2006, New York, NY, USA, March 4-7, 2006, Proceedings},
  pages 265--284, 2006.

\bibitem{DBLP:conf/icalp/Dwork06}
Cynthia Dwork.
\newblock Differential privacy.
\newblock In {\em ICALP (2)}, pages 1--12, 2006.

\bibitem{DBLP:journals/fttcs/DworkR14}
Cynthia Dwork and Aaron Roth.
\newblock The algorithmic foundations of differential privacy.
\newblock {\em Foundations and Trends in Theoretical Computer Science},
  9(3-4):211--407, 2014.

\bibitem{DBLP:books/sp/17/Vadhan17}
Salil~P. Vadhan.
\newblock The complexity of differential privacy.
\newblock In {\em Tutorials on the Foundations of Cryptography.}, pages
  347--450. 2017.

\bibitem{DBLP:conf/ccs/ErlingssonPK14}
{\'{U}}lfar Erlingsson, Vasyl Pihur, and Aleksandra Korolova.
\newblock {RAPPOR:} randomized aggregatable privacy-preserving ordinal
  response.
\newblock In {\em Proceedings of the 2014 {ACM} {SIGSAC} Conference on Computer
  and Communications Security, Scottsdale, AZ, USA, November 3-7, 2014}, pages
  1054--1067, 2014.

\bibitem{DBLP:journals/corr/abs-1709-02753}
Jun Tang, Aleksandra Korolova, Xiaolong Bai, Xueqiang Wang, and XiaoFeng Wang.
\newblock Privacy loss in apple's implementation of differential privacy on
  macos 10.12.
\newblock {\em CoRR}, abs/1709.02753, 2017.

\bibitem{DBLP:journals/ijufks/Sweene02}
Latanya Sweeney.
\newblock k-anonymity: {A} model for protecting privacy.
\newblock {\em International Journal of Uncertainty, Fuzziness and
  Knowledge-Based Systems}, 10(5):557--570, 2002.

\bibitem{DBLP:conf/icde/MachanavajjhalaGKV06}
Ashwin Machanavajjhala, Johannes Gehrke, Daniel Kifer, and Muthuramakrishnan
  Venkitasubramaniam.
\newblock l-diversity: Privacy beyond k-anonymity.
\newblock In {\em Proceedings of the 22nd International Conference on Data
  Engineering, {ICDE} 2006, 3-8 April 2006, Atlanta, GA, {USA}}, page~24, 2006.

\bibitem{DBLP:conf/icde/LiLV07}
Ninghui Li, Tiancheng Li, and Suresh Venkatasubramanian.
\newblock t-closeness: Privacy beyond k-anonymity and l-diversity.
\newblock In {\em Proceedings of the 23rd International Conference on Data
  Engineering, {ICDE} 2007, The Marmara Hotel, Istanbul, Turkey, April 15-20,
  2007}, pages 106--115, 2007.

\bibitem{Shannon1948}
Claude~Elwood Shannon.
\newblock A mathematical theory of communication.
\newblock {\em The Bell System Technical Journal}, 27(3):379--423, 7 1948.

\bibitem{sep-information-semantic}
Luciano Floridi.
\newblock Semantic conceptions of information.
\newblock In Edward~N. Zalta, editor, {\em The Stanford Encyclopedia of
  Philosophy}. Metaphysics Research Lab, Stanford University, summer 2019
  edition, 2019.

\bibitem{DBLP:conf/sigmod/KiferM11}
Daniel Kifer and Ashwin Machanavajjhala.
\newblock No free lunch in data privacy.
\newblock In {\em Proceedings of the {ACM} {SIGMOD} International Conference on
  Management of Data, {SIGMOD} 2011, Athens, Greece, June 12-16, 2011}, pages
  193--204, 2011.

\bibitem{DBLP:conf/crypto/GehrkeHLP12}
Johannes Gehrke, Michael Hay, Edward Lui, and Rafael Pass.
\newblock Crowd-blending privacy.
\newblock In {\em Advances in Cryptology - {CRYPTO} 2012 - 32nd Annual
  Cryptology Conference, Santa Barbara, CA, USA, August 19-23, 2012.
  Proceedings}, pages 479--496, 2012.

\bibitem{DBLP:conf/tcc/KasiviswanathanNRS13}
Shiva~Prasad Kasiviswanathan, Kobbi Nissim, Sofya Raskhodnikova, and Adam~D.
  Smith.
\newblock Analyzing graphs with node differential privacy.
\newblock In {\em {TCC}}, pages 457--476, 2013.

\bibitem{DBLP:conf/sigmod/ChenZ13}
Shixi Chen and Shuigeng Zhou.
\newblock Recursive mechanism: towards node differential privacy and
  unrestricted joins.
\newblock In {\em Proceedings of the {ACM} {SIGMOD} International Conference on
  Management of Data, {SIGMOD} 2013, New York, NY, USA, June 22-27, 2013},
  pages 653--664, 2013.

\bibitem{DBLP:journals/tods/KiferM14}
Daniel Kifer and Ashwin Machanavajjhala.
\newblock Pufferfish: {A} framework for mathematical privacy definitions.
\newblock {\em {ACM} Trans. Database Syst.}, 39(1):3, 2014.

\bibitem{DBLP:conf/focs/BassilyGKS13}
Raef Bassily, Adam Groce, Jonathan Katz, and Adam~D. Smith.
\newblock Coupled-worlds privacy: Exploiting adversarial uncertainty in
  statistical data privacy.
\newblock In {\em 54th Annual {IEEE} Symposium on Foundations of Computer
  Science, {FOCS} 2013, 26-29 October, 2013, Berkeley, CA, {USA}}, pages
  439--448, 2013.

\bibitem{DBLP:conf/tcc/GehrkeLP11}
Johannes Gehrke, Edward Lui, and Rafael Pass.
\newblock Towards privacy for social networks: {A} zero-knowledge based
  definition of privacy.
\newblock In {\em Theory of Cryptography - 8th Theory of Cryptography
  Conference, {TCC} 2011, Providence, RI, USA, March 28-30, 2011. Proceedings},
  pages 432--449, 2011.

\bibitem{DBLP:conf/innovations/GhoshK16}
Arpita Ghosh and Robert Kleinberg.
\newblock Inferential privacy guarantees for differentially private mechanisms.
\newblock In {\em Proceedings of the 2017 {ACM} Conference on Innovations in
  Theoretical Computer Science, Berkeley, USA, January 9-11, 2017}, pages~--,
  2017.

\bibitem{wu-he-xia2018}
Genqiang Wu, Xianyao Xia, and Yeping He.
\newblock Information theory of data privacy.
\newblock {\em CoRR}, abs/1703.07474v4, 2018.
  \url{https://arxiv.org/abs/1703.07474v4}.

\bibitem{6769090}
Claude~Elwood Shannon.
\newblock Communication theory of secrecy systems.
\newblock {\em The Bell System Technical Journal}, 28(4):656--715, 1949.

\bibitem{DBLP:journals/corr/abs-1907-09311}
Genqiang Wu.
\newblock On the information privacy model: the group and composition privacy.
\newblock {\em CoRR}, abs/1907.09311, 2019.

\bibitem{DBLP:books/daglib/0016881}
Thomas~M. Cover and Joy~A. Thomas.
\newblock {\em Elements of information theory}.
\newblock Tsinghua University Press, 2003.

\bibitem{DBLP:conf/nips/DworkFHPRR15}
Cynthia Dwork, Vitaly Feldman, Moritz Hardt, Toniann Pitassi, Omer Reingold,
  and Aaron Roth.
\newblock Generalization in adaptive data analysis and holdout reuse.
\newblock In {\em Advances in Neural Information Processing Systems 28: Annual
  Conference on Neural Information Processing Systems 2015, December 7-12,
  2015, Montreal, Quebec, Canada}, pages 2350--2358, 2015.

\bibitem{DBLP:conf/focs/RogersRST16}
Ryan~M. Rogers, Aaron Roth, Adam~D. Smith, and Om~Thakkar.
\newblock Max-information, differential privacy, and post-selection hypothesis
  testing.
\newblock In {\em {IEEE} 57th Annual Symposium on Foundations of Computer
  Science, {FOCS} 2016, 9-11 October 2016, Hyatt Regency, New Brunswick, New
  Jersey, {USA}}, pages 487--494, 2016.

\bibitem{DBLP:conf/sigmod/HayMMCZ16}
Michael Hay, Ashwin Machanavajjhala, Gerome Miklau, Yan Chen, and Dan Zhang.
\newblock Principled evaluation of differentially private algorithms using
  dpbench.
\newblock In {\em Proceedings of the 2016 International Conference on
  Management of Data, {SIGMOD} Conference 2016, San Francisco, CA, USA, June 26
  - July 01, 2016}, pages 139--154, 2016.

\bibitem{DBLP:journals/isci/Soria-ComasD13}
Jordi Soria{-}Comas and Josep Domingo{-}Ferrer.
\newblock Optimal data-independent noise for differential privacy.
\newblock {\em Inf. Sci.}, 250:200--214, 2013.

\bibitem{DBLP:conf/focs/McSherryT07}
Frank McSherry and Kunal Talwar.
\newblock Mechanism design via differential privacy.
\newblock In {\em FOCS}, pages 94--103, 2007.

\bibitem{wu-he-xia2017-b}
Genqiang Wu, Xianyao Xia, and Yeping He.
\newblock Analytic theory to differential privacy.
\newblock {\em CoRR}, abs/1702.02721, 2018.
  \url{https://arxiv.org/abs/1702.02721}.

\bibitem{athreya2006measure}
Krishna~B Athreya and Soumendra~N Lahiri.
\newblock {\em Measure theory and probability theory}.
\newblock Springer Science \& Business Media, 2006.

\end{thebibliography}

\end{document}